\documentclass[11pt]{amsart}
\usepackage[latin1]{inputenc}
\usepackage{amssymb,epsfig}
\usepackage{amsmath, amsthm}
\usepackage{a4wide}
\usepackage{amsfonts}
\usepackage{float}
\usepackage{cite}
\usepackage{pdfsync}
\usepackage{color}
\usepackage{enumerate}

\def\E{{\mathbb E}}

\newtheorem{theorem}{Theorem}
\newtheorem{lemma}{Lemma}
\newtheorem{corollary}{Corollary}

\newtheorem{remark}{Remark}

\title{Central Limit Theorem for the Elephant Random Walk}
\author{Cristian F. Coletti (1), Renato Gava (2) and Gunter M. Sch\"utz (3)}

\date{\today}

\address{
\newline 
(1) UFABC - Centro de Matem\'atica, Computa\c{c}\~ao e Cogni\c{c}\~ao.
\newline
Avenida dos Estados, 5001, Santo Andr\'e - S\~ao Paulo, Brasil
\newline
e-mail:  \rm \texttt{cristian.coletti@ufabc.edu.br}
\newline 
\newline
(2) UFSCAR - Departamento de Estat\'{\i}stica.
\newline  Rodovia Washington Luiz, Km 235, CEP 13565-905, S\~ao Carlos, Brasil
\newline
e-mail:  \rm \texttt{gava@ufscar.br} 
\newline
\newline
Institute of Complex Systems II
\newline
(3) Forschungszentrum J\"ulich, 52425 J\"ulich, Germany
\newline
e-mail:  \rm \texttt{g.schuetz@fz-juelich.de}
}

\begin{document}

\maketitle

Abstract: We study the so-called elephant random walk which is a non-Markovian
discrete-time random walk on $\mathbb{Z}$ with unbounded memory which exhibits
a phase transition from diffusive to superdiffusive behaviour. We prove a law of large
numbers and a central limit theorem. Remarkably the central limit theorem applies
not only to the diffusive regime but also to the phase transition
point which is superdiffusive. Inside the superdiffusive regime the 
ERW converges to a non-degenerate random variable which is not normal. 
We also obtain explicit expressions for the correlations of increments of the
ERW.


\section{Introduction}\label{intro}

In this paper we consider a model presented in \cite{ST}, the so-called elephant random walk
(ERW). This is a discrete-time random walk $X_n$ on $\mathbb{Z}$ with unbounded memory 
whose random increments at each time step depend on the whole history of the process. 
This model a rare example of a non-Markovian process
for which exact results on the moments of $X_n$ are available \cite{ST,Para06,daSi13}.
Also some large deviation results have been obtained \cite{Harr09}.
Significantly, there is a phase transition from diffusive to superdiffusive behaviour
as a function of a memory parameter $p$ (defined below) at a critical value $p_c$. 
A modified ERW was shown to exhibit also subdiffusion \cite{Harb14}. Recently a
surprising connection between the ERW and bond percolation on random recursive trees
has been found \cite{Kuer16}. 

An open question that has remained is the actual probability distribution of $X_n$ in 
the original ERW. Initially it was suggested that on large scales 
the probability density is Gaussian, obeying a Fokker-Planck equation with
time-dependent drift term \cite{ST}. Later it was conjectured that the distribution
is Gaussian in the diffusive regime, but not in the superdiffusive regime 
\cite{Para06,daSi13}. Here we address this problem. We prove rigorously a law of large numbers
that is valid for any non-trivial value of the memory parameter and a central limit theorem
that is valid in the diffusive regime as well as at the critical value $p_c$ where
the model is already superdiffusive. Inside the superdiffusive regime $X_n$ is shown to be
a non-degenerate but non-normal random variable. Some further insight is obtained
by a precise result on the correlations between the random increments.

We would like to point out that the study of the limiting behavior of the ERW may be related to the
study of limit theorems for a class of correlated Bernoulli processes.
We refer the reader interested in this approach to the work of Lan Wu et al. \cite{WuQiYang} and references therein.
In this context, it is worth mentioning the recent work of Gonz\'alez-Navarrete and Lambert \cite{GNL} where the authors propose an approach to construct and study
Bernoulli sequences combining dependence and independence periods.

The paper is organized as follows. In the next section we define the process and present
our main results which are proved in Sec. 4. In Sec. 3 we derive some auxiliary results 
that are used in the proofs of Sec. 4. We mention that upon completion of this work
we were made aware of related results on the ERW by Baur and Bertoin \cite{Baur16}.

\section{Definition of the ERW and main results}

The ERW is defined as follows. The walk starts at a specific point $X_0$ 
at time $n=0$. At each discrete time step the elephant moves one step to the right or to left respectively, so 
\begin{equation}
X_{n+1} = X_n + \eta_{n+1}
\end{equation}
where $\eta_{n+1} = \pm 1$ is a random variable. The memory consists of the set of random variables $\eta_{n^{\prime}}$
at previous time steps which the elephant remembers as follows:

\noindent $(D_1)$ At time $n+1$ a number $n^{\prime}$ from the set $\{1,2, \ldots , n\}$ is chosen at random with probability $1/n$.

\noindent $(D_2)$ $\eta_{n+1}$ is determined stochastically by the rule 
\begin{align*}
\eta_{n+1} = \eta_n^{\prime} \, \text{ with probability } \, p \, \text{ and } \, \eta_{n+1} = -\eta_n^{\prime} \, \text{ with probability } 1-p.
\end{align*}

\noindent $(D_3)$ The elephant starting at $X_0$ moves to the right with probability $q$ and to the left with probability $1-q$, i.e.,
\begin{align*}
\eta_{1} = 1  \, \text{ with probability } \, q \, \text{ and } \, \eta_{1} = -1 \, \text{ with probability } 1-q.
\end{align*}

\noindent It is obvious from the definition that
\begin{equation}
X_n = X_0 + \sum_{k=1}^n \eta_k.
\end{equation}

From now on we consider $X_0=0$. Therefore, $X_n = \displaystyle \sum_{k=1}^n \eta_k$.
A simple computation yields 
\begin{equation}\label{conditional}
\mathbb{P}[\eta_{n+1} = \eta|\eta_1, \ldots, \eta_n] = \frac{1}{2n} \sum_{k=1}^n \left[1+\left(2p-1\right)\eta_k \eta\right] \ \mbox{for} \ n \geq 1,
\end{equation}
where $\eta = \pm 1$. For $n=0$ we get in accordance with rule $(D_3)$
\begin{equation} 
\mathbb{P}[\eta_1 = \eta] = \frac{1}{2}  \left[1+\left(2q-1\right)\eta\right]
\end{equation}
and
\begin{equation}
\mathbb{E}[\eta_1] = 2q- 1.
\end{equation}
The conditional expectation of the increment $\eta_{n+1}$ given its previous history is given by
\begin{equation}\label{conditional2} 
\mathbb{E}[\eta_{n+1} |\eta_1, \ldots, \eta_n] = (2p-1) \frac{X_n}{n} \ \mbox{for} \ n \geq 2.
\end{equation}

\begin{remark}
Sch\"utz and Trimper \cite{ST} showed that
\begin{eqnarray} \label{mean}
\mathbb{E}[X_n] = (2q-1) \frac{\Gamma(n+(2p-1))}{\Gamma(2p)\Gamma(n)} \sim \frac{2q-1}{\Gamma(2p)} n^{2p-1} \ \mbox{for} \ n \gg 1.
\end{eqnarray}
Therefore,
\begin{eqnarray}\label{asymptoticexpectation}
\mathbb{E}[\eta_{n+1}] &\sim&  \frac{\left(2p-1\right) \left(2q-1\right)}{\Gamma(2p)} n^{2(p-1)}. 
\end{eqnarray}
\end{remark}



Now we state the main results. 

\begin{theorem}\label{lln}
Let $(X_n)_{n \geq1}$ be the elephant random walk. Then
\begin{eqnarray}
\lim_{n \rightarrow \infty} \frac{X_n - \mathbb{E}[X_n]}{n} = 0 \ \mbox{a.s.} \nonumber
\end{eqnarray}
for any value of $q$ and $p \in [0,1)$.
\end{theorem}

\begin{remark}
If $p=1$, which is not covered by the law of large numbers of Theorem \ref{lln}, 
the ERW is trivial since by definition of the process one then has $\eta_n=\eta_1$
for all $n\geq 1$. Hence $X_n / n = \eta_1$ reduces to a binary random variable. 
\end{remark}

Interestingly, as the next two theorems demonstrate, the non-Markovian nature of the
ERW is somewhat disguised for $p \leq 3/4$, but shows up for $p>3/4$.

\begin{theorem}\label{clt}
Let $(X_n)_{n \geq 1}$ be the elephant random walk and let $p \leq 3/4$. 
(a) If $p < 3/4$, then 
\begin{align}
\dfrac{X_n - \dfrac{2q - 1}{\Gamma(2p)}n^{2p -1}}{\sqrt{\dfrac{n}{3 - 4p}}} \xrightarrow{\mathcal{D}} \mathcal{N}(0,1). \nonumber
\end{align}
(b) If $p = 3/4$, then
\begin{align}
 \dfrac{X_n - \dfrac{2q - 1}{\Gamma(3/2)}n^{1/2}}{\sqrt{n \ln n}} \xrightarrow{\mathcal{D}} \mathcal{N}(0,1). \nonumber
\end{align}
\end{theorem}

\begin{theorem}\label{superdiff}
Let $(X_n)_{n \geq 1}$ be the elephant random walk. If $3/4 < p \leq 1$, then
\begin{eqnarray}
\frac{X_n}{n^{2p -1} \Gamma(2p)^{-1}} - (2q -1) \to M \text{ a.s. } \nonumber, 
\end{eqnarray}
where $M$ is a non-degenerate mean zero random variable, but not a
normal random variable.
\end{theorem}

The absence of convergence to a normal r.v. for $p>3/4$ suggests that cross-correlations
between the increments $\eta_n$ are in some sense ``too strong''. The following theorem
quantifies these correlations.

\begin{theorem}\label{corr}
Let $(X_n)_{n \geq 1}$ be the elephant random walk with $q=1/2$ and $0 \leq p \leq 1$
and define
\begin{eqnarray}
F(n) := \frac{2p-1}{3-4p }\left(\frac{2(1-p)}{n} - \frac{(2p-1)
\Gamma(n+4p-3)}{\Gamma(4p-2)n!}\right) . \nonumber
\end{eqnarray}
Then
$\mathbb{E} [\eta_n] =0$ and
\begin{eqnarray}
\mathbb{E} [\eta_n \eta_{n+k}] = \frac{n! \Gamma(2p+n+k-2)}
{(n+k-1)! \Gamma(2p-1+n)} F(n) \nonumber
\end{eqnarray}
for all $n\geq 1$ and $k\geq 1$.
\end{theorem}

\begin{remark}
For $p=3/4$ one has
\begin{equation}\label{2-20a}
F(n) =\frac{1}{2n}\left(1 + \frac{1}{2} \sum_{m=1}^{n-1}
\frac{1}{m} \right).
\end{equation}
For $k,n \to\infty$ with $x:=k/n$ fixed the correlation function 
has the form
\begin{equation}\label{2-21}
\mathbb{E} [\eta_n \eta_{n+k}] \sim (1+x)^{-2(1-p)} F(n)
\end{equation}
where $F(n)$ decays
algebraically except at the transition point $p=3/4$. In the various regimes 
Stirling's formula for the $Gamma$-function gives
\begin{equation}\label{2-22}
F(n) \sim \left\{ \begin{array}{ll}
\displaystyle \frac{2(2p-1)(1-p)}{3-4p} n^{-1} & p < 3/4 \\[2mm]
\displaystyle \frac{ \ln{n}}{4n} & p = 3/4 \\[2mm]
\displaystyle  \frac{(2p-1)^2}{(4p-3)\Gamma(4p-2)}
n^{-4(1-p)} & p > 3/4.
\end{array} \right.
\end{equation}
It appears that the transition at $p=3/4$ is driven by the strength of the
correlations at time step $n$ rather than by their decay with the time-lag $k$
between increments.
\end{remark}

\section{Auxiliary results}

In this section we present some results that will be used in the section \ref{proof}
where we prove  Theorems \ref{lln} - \ref{corr}.  

Put 
\begin{align*}
a_1 = 1 \, \text{ and } a_n = \prod_{j=1}^{n-1} \left(1+\frac{(2p-1)}{j}\right) \, \text{ for } n \geq 2.
\end{align*}
Define the filtration $\mathcal{F}_n=\sigma(\eta_1, \ldots, \eta_n)$ and $M_n = \frac{X_n - \mathbb{E}[X_n]}{a_n}$ for $n \geq 1$. 
We claim that $\{ M_n \}_{n \geq 1}$ is a martingale with respect to $\{ \mathcal{F}_n \}_{n \geq 1}$, for
\begin{eqnarray}
\mathbb{E}[M_{n+1}| \mathcal{F}_n] &=& \frac{(X_n - \mathbb{E}[X_n])}{a_{n+1}} + \frac{\mathbb{E}[\eta_{n+1}|\mathcal{F}_n] - \mathbb{E}[\eta_{n+1}]}{a_{n+1}} \nonumber \\
&=& \frac{(X_n - \mathbb{E}[X_n])}{a_{n+1}} + \frac{(2p-1) \frac{X_n}{n} -(2p-1) \frac{\mathbb{E}[X_n]}{n}}{a_{n+1}} \nonumber \\
&=& \frac{(X_n - \mathbb{E}[X_n])}{a_{n+1}} + \frac{\frac{(2p-1)}{n}(X_n - \mathbb{E}[X_n])}{a_{n+1}} \nonumber \\
&=& (X_n - \mathbb{E}[X_n]) \frac{\left(1+\frac{(2p-1)}{n}\right)}{a_{n+1}} \nonumber \\
&=& M_n .\nonumber
\end{eqnarray}

Before proving the next lemma, let us make an important remark on $a_n$. Using the gamma function, we can 
rewrite $a_n$ in the following way
\begin{align}\label{approx}
a_{n} = \frac{\Gamma(n + 2p -1)}{\Gamma(n)\Gamma(2p)} \sim \frac{n^{2p -1}}{\Gamma(2p)} \, \text{ as } n \to \infty.
\end{align}
Notice that $a_n \to \infty$ as $n \to \infty$ if $p > 1/2$, $a_n =1$ for $n \geq 1$ if $p = 1/2$, and $a_n \to 0 $ if $p < 1/2$. 

\begin{lemma}\label{diverges}
The series $\displaystyle \sum_{n=1}^{\infty} \frac{1}{a^2_n}$ converges if and only if $p > 3/4$.
\end{lemma}

\begin{proof}
Let $b_n = \frac{1}{a^2_n}$ and assume that $p\neq 3/4$. Then,
\begin{eqnarray}
n \left(1 - \frac{b_{n}}{b_{n+1}}\right) &=& \frac{2 (2p-1)+\frac{(2p-1)^2}{n}}{1 + 2 \frac{(2p-1)}{n}+\frac{(2p-1)^2}{n^2}}
\end{eqnarray}
Therefore, 
\begin{eqnarray}
R:= \displaystyle \lim_{n \rightarrow +\infty} n \left(1 - \frac{b_{n}}{b_{n+1}}\right) = 4p-2. \nonumber 
\end{eqnarray}

By the Raabe criteria the series is convergent if $R > 1$ and it is divergent if $R < 1$. Note that $R > 1$ if and 
only if $p > \frac{3}{4}$ and $R < 1$ if and only if $p < \frac{3}{4}$.

Assume now that $p=\frac{3}{4}$. Since $2p -1 = 1/2 > 0$, it follows by (\ref{approx}) that 
\begin{align}\label{pigualtresquarto}
a_n = \prod_{j=1}^{n-1} \left(1+\frac{1/2}{j}\right) \sim \frac{n^{1/2}}{\sqrt{\pi/4}}.
\end{align}
Thence $1/a_n^{2} \geq \frac{1}{2n}$ for $n$ large enough and $ \displaystyle \sum_{n=2}^{\infty} \frac{1}{a^2_n}$ 
diverges.
\end{proof}

Let $(D_n)_{n \geq1}$ be the martingale differences defined by $D_1 = M_1$ and, for $n \geq 2, D_n = M_n - M_{n-1}$.
Observe that
\begin{eqnarray}\label{Dj}
D_j &=&  \frac{X_j - \mathbb{E}[X_j]}{a_j} - \frac{X_{j-1} - \mathbb{E}[X_{j-1}]}{a_{j-1}} \nonumber \\
&=& \frac{\eta_j - \mathbb{E}[\eta_j]}{a_j} - \frac{\left(X_{j-1} - \mathbb{E}[X_{j-1}]\right)}{j-1} \frac{2p-1}{a_j} . 
\end{eqnarray}
Furthermore, since the increments $\eta_j$'s are uniformly bounded, it is trivial to see that 
\begin{align}\label{4aj}
|D_j| \leq \frac{4}{a_j}.
\end{align}

We state without proof the Kronecker lemma.
\begin{lemma}\label{Kro}
Let $(x_n)_{n \geq 1}$ and $(b_n)_{n \geq 1}$ be sequences of real numbers such that $0 < b_n \nearrow +\infty$. 
If $\sum_k x_k/b_k$ converges, then $\displaystyle \lim_{n \rightarrow +\infty} \frac{\sum_{k=1}^n x_k}{b_n} = 0$.
\end{lemma}

\section{Proofs}\label{proof}

\subsection{Strong law of large numbers}

\begin{proof}[Proof of Theorem  \ref{lln}]
First observe that 
\begin{eqnarray}
\label{an}
\frac{a_n}{n} &=& \frac{1}{n} \displaystyle \prod_{j=1}^{n-1} \left(1+\frac{2p-1}{j}\right) \nonumber \\
&=& \displaystyle \prod_{j=1}^{n-1} \frac{j+2p-1}{j+1} \nonumber \\
&=& \displaystyle \prod_{j=1}^{n-1} \left( 1 - \frac{2(1 -p)}{j+1} \right)
\end{eqnarray}
where $0 \leq \frac{j+2p-1}{j+1} \leq 1$. Observe that $\displaystyle \sum_{j=1}^{\infty} \frac{2(1 -p)}{j+1} = \infty$ 
implies that $\displaystyle \lim_{n \rightarrow +\infty} \frac{a_n}{n} = 0$. Furthermore, it follows from (\ref{an}) 
and the fact that $0 \leq \frac{j+2p-1}{j+1} \leq 1$ that $(a_n/n)_{n \geq 1}$ is a non-increasing sequence.

Define $N_j = \frac{a_j}{j} D_j$. By (\ref{4aj}) we may conclude that $(N_j)_{j \geq 1}$ is a sequence of martingale 
differences such that 
\begin{eqnarray}
\sum_{j=1}^{+\infty} \mathbb{E}[N^2_j | \mathcal{F}_{j-1}] \leq  \sum_{j=1}^{+\infty} \frac{4^2}{j^2} < + \infty \text{ a.s.}. \nonumber
\end{eqnarray}
Now, Theorem 2.17 in \cite{HH} implies that $\displaystyle \sum_{j=1}^{+\infty} N_j$ converges almost surely. 
Since $n/a_n \nearrow \infty$, a direct application of Lemma \ref{Kro} gives
\begin{eqnarray*}
\frac{a_n}{n} \displaystyle \sum_{j=1}^n D_j = \frac{a_n}{n} M_n \to 0 \, \text{ almost surely. }
\end{eqnarray*}
Next, note that 
\begin{eqnarray}
\frac{X_n - \mathbb{E}[X_n]}{n} &=& \frac{a_n}{n} M_n, \nonumber 
\end{eqnarray}
i.e. the law of large number holds for $X_n$.
\end{proof}

\begin{corollary}
Let $(X_n)_{n \geq 1}$ be the elephant random walk. Then
\begin{eqnarray}
\lim_{n \rightarrow \infty} \frac{X_n}{n} = 0 \ \mbox{a.s.} \nonumber
\end{eqnarray}
for any value of $q$ and $p \in [0,1)$.
\end{corollary}

\begin{proof}
It follows from (\ref{mean}) that $\frac{\mathbb{E}[X_n]}{n} \sim \frac{2q-1}{\Gamma(2p)} n^{2(p-1)}$,
which in turn goes to 0 as $n \to \infty$, and the claim follows.
\end{proof}

\subsection{Central limit theorem}

\begin{proof}[Proof of Theorem \ref{clt}]

Define 
\begin{align*}
s^2_1=q(1-q) \, \text{ and } s^2_n := q(1-q) + \sum_{j=2}^n \frac{1}{a^2_j} \, \text{ for } n \geq 2.
\end{align*}

Before proving the results we need to say that both claims in Theorem \ref{clt} amounts to show that   
\begin{eqnarray}\label{ansn}
\frac{X_n - \mathbb{E}[X_n]}{a_n s_n} \xrightarrow{\mathcal{D}} \mathcal{N}(0,1).
\end{eqnarray}
Indeed, we will show that (\ref{ansn}) holds. First recall that $\mathbb{E}(X_n) \sim  (2q - 1) n^{2p -1}/\Gamma(2p)$, by (\ref{mean}).
Combining Lemma \ref{diverges} and (\ref{approx}) we get
\begin{align*}
& s_n^{2} \sim  \Gamma(2p)^{2} \frac{n^{3 - 4p}}{3 - 4p}, \text{ if } p < 3/4 \\
& s_n^{2} \sim \Gamma(3/2)^{2} \ln n, \text{ if } p =3/4
\end{align*}
which in turn implies that $a_n s_n \sim \sqrt{n \ln n}$, if $p = 3/4$, and $a_n s_n \sim \sqrt{n/(3 - 4p)}$, 
if $p < 3/4$.

We now turn our attention to the proof of (\ref{ansn}). We will verify the two conditions of Corolary 3.1 
of \cite{HH} in order to get our result.

We begin by checking the conditional Lindeberg condition. Let $D_{nj} = \frac{D_j}{s_n}$ for 
$1 \leq j \leq n$. Given $\varepsilon >0$, we need to prove that
\begin{align}\label{tozero}
\sum_{j = 1}^{n}\mathbb{E}(D^{2}_{nj} \mathbb{I} (|D_{nj}| > \varepsilon) | \mathcal{F}_{j-1}) \to 0  \text{ a.s.}.
\end{align}
If $p \in [1/2, 3/4]$, then $a_n \geq 1$ and $s_n \to \infty$. Next note that 
$\lbrace |D_{nj}| > \varepsilon \rbrace \subset \lbrace \frac{4}{s_n} > \varepsilon \rbrace $,
but the last set is empty for $n$ large enough, so (\ref{tozero}) holds.
If $p < 1/2$, observe that $\lim_{n \to \infty }a_n \, s_n  = \infty$ and $a_j^{-1} \leq a_n^{-1}$ for 
$j= 1, \ldots, n$. Then it is easy to see that 
$\lbrace |D_{nj}| > \varepsilon \rbrace \subset \lbrace \frac{4}{a_n s_n} > \varepsilon \rbrace $,
and again the latter set is empty for $n$ sufficiently large.

Next we  check the conditional variance condition. Since $ \frac{\sum_{j=1}^n D_j}{s_n} = \frac{S_n - \mathbb{E}[S_n]}{a_n s_n}$, 
we must estimate $\mathbb{E}[D_j^2 | \mathcal{F}_{j-1}]$. Observe that
\begin{eqnarray}
D^2_j 
&=& \frac{1}{a^2_j}\left(\left(1 - 2 \mathbb{E}[\eta_j] \eta_j + \mathbb{E}^2[\eta_j]\right) + \left(\frac{\left(X_{j-1} - \mathbb{E}[X_{j-1}]\right)}{j-1}\right)^2 (2p-1)^2\right) \nonumber \\
&-& \frac{2}{a^2_j}\left(\left(\eta_j - \mathbb{E}[\eta_j]\right) \left(\frac{\left(X_{j-1} - \mathbb{E}[X_{j-1}]\right)}{j-1}\right) (2p-1)\right), \nonumber 
\end{eqnarray}
for $\eta^2_j = 1$.

It follows from the law of large numbers for $X_n$ that 
\[
\mathbb{E}\left[\frac{1}{a^2_j}\left(\frac{\left(X_{j-1} - \mathbb{E}[X_{j-1}]\right)}{j-1}\right)^2 (2p-1)^2 | \mathcal{F}_{j-1} \right]
\]

and 

\[
\mathbb{E}\left[\frac{1}{a^2_j}\left(\left(\eta_j - \mathbb{E}[\eta_j]\right) \left(\frac{\left(X_{j-1} - \mathbb{E}[X_{j-1}]\right)}{j-1}\right) (2p-1)\right) | \mathcal{F}_{j-1} \right]
\]
are $\mbox{o}\left(\frac{1}{a^2_j}\right)$. Therefore, in order to obtain the assimptotic behavior of $\mathbb{E}\left[D^2_j | \mathcal{F}_{j-1}\right]$, 
we need to study the behavior of $\frac{1}{a^2_j}\left(\mathbb{E}\left[\left(1 - 2 \mathbb{E}[\eta_j] \eta_j + \mathbb{E}^2[\eta_j]\right) | \mathcal{F}_{j-1}\right]\right)$.


It follows from (\ref{conditional}) and (\ref{asymptoticexpectation})  that, for $j \geq 2$,

\begin{eqnarray}
\frac{1}{a^2_j} \mathbb{E}\left[\left(1 - 2 \mathbb{E}[\eta_j] \eta_j + \mathbb{E}^2[\eta_j]\right) | \mathcal{F}_{j-1}\right] &\sim& \frac{1}{a^2_j}\left(1-2 \ \frac{(2p-1)(2q-1)}{\Gamma(2p)(j-1)^{2(1-p)}} \frac{(2p-1) X_{j-1}}{j-1}\right) \nonumber \\
&+& \frac{1}{a^2_j} \left(\frac{(2p-1)(2q-1)}{\Gamma(2p)(j-1)^{2(1-p)}}\right)^2 \nonumber \\
&=& \frac{1}{a^2_j} + \ \mbox{o}\left(\frac{1}{a^2_j}\right). \nonumber
\end{eqnarray}
Therefore, for $j \geq 2$ we have that
\begin{eqnarray}\label{Conditionaldj}
\mathbb{E}\left[D^2_j | \mathcal{F}_{j-1}\right] = \frac{1}{a^2_j} + \ \mbox{o}\left(\frac{1}{a^2_j}\right).
\end{eqnarray}
Finally, for $j=1$, we have that
\begin{eqnarray}\label{Conditionald1}
\mathbb{E}\left[D^2_1 | \mathcal{F}_0\right] &=& \frac{1}{a^2_1} q (1-q),
\end{eqnarray}
where $a_1 = 1$. Recall that $s^2_1=q(1-q)$ and $s^2_n := q(1-q) + \sum_{j=2}^n \frac{1}{a^2_j}$ for 
$n \geq 2$. From Lemma \ref{diverges} we have that $s_n \to \infty$ as $n \to \infty$ if and only 
if $p \leq 3/4$. From this fact, (\ref{Conditionaldj}), (\ref{Conditionald1}) and Theorem \ref{lln} we may conclude that
\begin{eqnarray}\label{as}
\frac{1}{s_n^2} \displaystyle \sum_{j=1}^n \mathbb{E}[D_j^2 | \mathcal{F}_{j-1}] \xrightarrow{a.s} 1.
\end{eqnarray}
Therefore, by Corolary 3.1. in \cite{HH} we conclude that
\begin{eqnarray}
\displaystyle \sum_{j=1}^n D_{nj} = \frac{X_n - \mathbb{E}[X_n]}{a_n s_n} \xrightarrow{\mathcal{D}} \mathcal{N}(0,1). \nonumber
\end{eqnarray}
\end{proof}

\subsection{Almost sure convergence for $p > \frac{3}{4}$}

\begin{proof}[Proof of Theorem \ref{superdiff}]
First note that $\mathbb{E}(X_n) = (2q - 1)a_n$ and $a_n \sim n^{2p -1}/\Gamma(2p)$, by (\ref{mean}) and (\ref{approx}).
Lemma \ref{diverges} says that if $p > \frac{3}{4}$, then $\displaystyle \sum_{n=1}^{\infty} \frac{1}{a^2_n} < \infty$. 
For $|D_j| \leq \frac{4}{a_j}$ by (\ref{4aj}), then Theorem 12.1 in \cite{Wi} implies that
$M_n = \displaystyle \sum_{j=1}^n D_j = \frac{X_n - \mathbb{E}[X_n]}{a_n} \to M$ almost surely and in $\mathcal{L}^{2}$.
These results and $\mathbb{E}(M_n) = 0$ together imply that  
\begin{align*}
|\mathbb{E}(M)| = |\mathbb{E}(M - M_n)| \leq \mathbb{E}(|M - M_n|) \leq \mathbb{E}(|M - M_n|^{2})^{1/2} \to 0 \, \text{ as } n \to \infty. 
\end{align*}
Furthermore, the fact that $(M_n)_{n \geq 1}$ is a bounded martingale in $\mathcal{L}^{2}$ plus its a.s. convergence
to $M$ implies that
\begin{align*}
Var(M) = \lim_{n \to \infty}Var(M_n) = \sum_{j=1}^{\infty}\mathbb{E}(D_j^{2}) > 0.
\end{align*}
Then, $M$ is not a degenerate r.v.. The claim that $M$ is not normal follows from the
explicit form of the skewness and kurtosis of the distribution of $X_n$
obtained by exact computation in \cite{Para06}.
\end{proof}

\subsection{Correlations}

\begin{proof}[Proof of Theorem \ref{corr}]

To compute the correlation function $\E[\eta_{n+k}\eta_{n}]$
for $n\geq1$ and $k\geq 1$ we introduce
the auxiliary quantity $H_k = (n+k-1)
\E[\eta_{n+k}X_n]$. It follows from (\ref{conditional2}) that
\begin{eqnarray}
H_{k+1} = \left(1+\frac{2p-1}{n+k-1}\right) H_k, \quad H_1 = (2p-1) \E[X_n^2]
\end{eqnarray}
which yields by induction
\begin{eqnarray}
H_k = \frac{(n-1)! \Gamma(2p+n+k-2)}{(n+k-2)!
\Gamma(2p-1+n)} H_1.
\end{eqnarray}
Therefore
\begin{eqnarray}
\E[\eta_{n+k}X_n] & = & \frac{n! \Gamma(2p+n+k-2)}{(n+k-1)!
\Gamma(2p-1+n)} \E[\eta_{n+1}X_n] \nonumber \\
& = & (2p-1) \frac{(n-1)! \Gamma(2p+n+k-2)}{(n+k-1)!
\Gamma(2p-1+n)} \E[X_n^2].
\end{eqnarray}
Subtracting $\E[\eta_{n+k}X_{n-1}]$ yields $\E[\eta_{n+k}\eta_{n}]$ in the form
stated in the theorem with
\begin{eqnarray}
F(n) = \frac{2p-1}{n} \left[ \E[X_n^2] -
\left(1+\frac{2p-1}{n-1}\right) \E[X_{n-1}^2]\right].
\end{eqnarray}
With the explicitly known expression \cite{ST} 
\begin{eqnarray}
\E[X_n^2] = \frac{n}{4p-3}\left(
\frac{\Gamma(n+4p-2)}{\Gamma(4p-2)\Gamma(n+1)} - 1 \right)
\end{eqnarray}
we arrive at the desired result.
\end{proof}

\end{document}